\documentclass[letterpaper, 10pt, onecolumn]{ieeeconf}  

\IEEEoverridecommandlockouts                              
\usepackage{graphicx} 
\usepackage{mathptmx} 
\usepackage{times} 
\usepackage{amsmath} 
\usepackage{amssymb}  
\usepackage{datetime}
\usepackage{hyperref}
\newtheorem{lem}{Lemma}
\newtheorem{thm}{Theorem}

\def\<{\leqslant}           
\def\>{\geqslant}           

\def\d{\partial}

\def\Re{{\rm Re}}   

\def\mR{{\mathbb R}}    

\def\rT{\mathrm{T}}        
\def\diam{\mathrm{diam}}       

\def\bS{\mathbf{S}}

\def\[[[{[\![\![}
\def\]]]{]\!]\!]}

\def\re{{\rm e}}        
\def\rd{{\rm d}}        

\def\br{{\bf r}}
\def\x{\times}

\def\fR{{\mathfrak R}}

\def\mS{{\mathbb S}}

\def\argmin{\mathop\mathrm{argmin}}
\def\esssup{\mathop\mathrm{ess\, sup}}
\def\clos{\mathrm{clos\, }}
\def\interior{\mathrm{int}}

\def\eps{\epsilon}

\title{\LARGE \bf
Convergence to periodic regimes in nonlinear feedback systems with a strongly convex backlash$^*$}

\author{Igor G. Vladimirov$^\dagger$,
\qquad
Ian R. Petersen$^\dagger$%
\thanks{$^*$This work is supported by
the Australian Research Council under grant DP180101805.}
\thanks{$^\dagger$Research School of Electrical, Energy and Materials Engineering, College of Engineering and Computer Science,
Australian National University, Canberra, Acton, ACT 2601,
Australia, {\tt\scriptsize
igor.g.vladimirov@gmail.com, i.r.petersen@gmail.com.}
}
}

\begin{document}

\maketitle

\thispagestyle{empty}

\begin{abstract}
This paper considers a class of nonlinear systems consisting of a linear part with an external input and a nonlinear feedback with a backlash. Assuming that the latter is specified by a strongly convex set, we establish estimates for the Lyapunov exponents which quantify the rate of convergence of the system trajectories to a forced periodic regime when the input is a periodic function of time. These results employ enhanced  dissipation inequalities for differential inclusions with strongly convex sets, which were used previously for the Moreau sweeping process.
\end{abstract}

\section{\bf Introduction}
\label{sec:intro}

Existence and uniqueness of periodic regimes and the rate of convergence to them are well-known for stable linear systems subject to periodic inputs. Such systems may result from a stabilising linear feedback applied to an unstable linear plant. If the feedback is implemented using mechanical elements (for example, gears, levers, spring-damper units, or inerters \cite{S_2002}), it can be susceptible to backlashes. These effects make the resulting closed-loop system nonlinear and augment its state.

A mathematical model for the backlash is considered in the theory of ordinary differential equations with
hysteresis nonlinearities \cite{KP_1989} (see also \cite{BPRR_2006}). This model describes the backlash as a closed convex set whose spatial location is specified by an input, with the output being moved only on its contact with the boundary of the set. Such dynamics form a special yet practically important case of the Moreau sweeping process \cite{M_1971,M_1977} generated by a convex-set-valued  map. Although the resulting differential equation has a discontinuous right-hand side, the initial value problem is well-posed due to the convexity of the backlash set and the properties of normal cones and metric projections onto such  sets.

The presence of strong convexity \cite{FO_1981,GI_2017,I_2006,I_2020,KP_1989,P_1975,PB_2004}, when the supporting hyperplanes and the corresponding  half-spaces can be replaced with balls of bounded radius, leads to the enhancement of properties  which employ convexity. In particular, this yields improved (quadratic rather than linear) convergence rates  for the numerical solution of differential games \cite{I_1994,IP_1995}  with geometric constraints specified by strongly convex sets. The condition of strong convexity for the set-valued map in the Moreau process mentioned above leads to an exponentially fast convergence \cite{V_1997} to periodic solutions, provided the map is periodic and has no equilibria.  The corresponding Lyapunov exponent involves the arc length of the solution over the period, so that the more the periodic solution ``moves'', the ``faster'' it attracts the other trajectories. This property is an example of dissipativity which is of differential geometric nature rather than coming from energy dissipation (due to friction) in the context of mechanical engineering applications.

The present paper extends the results of \cite{V_1997} to a class of closed-loop systems which consist of a linear part (a plant governed by a linear ODE with constant coefficients and an external input as a forcing term) and a nonlinear feedback with a strongly convex backlash. Self-induced oscillations in an autonomous version of such systems (with no external input and without the strong convexity) were studied, for example, in \cite{VC_1995}.  Assuming that the linear system is stable in the absence of the backlash (as if the backlash set were reduced to a singleton), the trajectories  of the linearised system (more precisely, their tubular neighbourhoods) provide a localization for those of the nonlinear system. If the system is subject to a periodic input with a sufficiently large ``amplitude'', this leads to a strictly positive lower bound for the arc length of the backlash output path   over the period. Under the strong convexity condition, this bound gives rise to dissipation inequalities which involve an interplay between the (energy-related) dissipation in the linear subsystem, the geometric dissipativity of the strongly convex backlash, and the plant-backlash  coupling. In combination with the Gronwall-Bellman lemma,  these differential inequalities  lead to estimates for the Lyapunov exponents quantifying the rate of convergence to periodic regimes in the nonlinear system. The assumption of large amplitudes for periodic external inputs can be replaced  here with that of smallness of the strong convexity constant for the backlash set.

The paper is organised as follows.
Section~\ref{sec:sys}
specifies the class of closed-loop systems with a backlash in the feedback being considered.
Section~\ref{sec:stat}
discusses those initial conditions for the backlash output which remain stationary over a bounded time interval.
Section~\ref{sec:loc}
provides a tubular localization for trajectories of the nonlinear system about those of its linearization.
Section~\ref{sec:act}
establishes asymptotic bounds for the path length of the backlash output and its time derivative for periodic inputs  of large amplitudes.
Section~\ref{sec:diss} considers dissipation inequalities for the nonlinear system in the case of a strongly convex backlash set and obtains the rates of convergence to periodic trajectories.
Section~\ref{sec:conc}
summarizes the results and outlines further directions of research.
Appendices~A and B provide auxiliary lemmas on inward normal cones for convex and strongly convex sets and a spectral bound for a class of real symmetric matrices.

\section{\bf Nonlinear systems being considered}
\label{sec:sys}

We consider a nonlinear time invariant system consisting of a linear part
\begin{equation}
\label{xy}
  \dot{x} = Ax + B w + E z,
  \qquad
  y = Cx,
\end{equation}
where
$w$, $x$, $y$, $z$ are functions of time with values in $\mR^m$, $\mR^n$, $\mR^p$, $\mR^p$, respectively (and $A \in \mR^{n\x n}$, $B\in \mR^{n\x m}$, $C\in \mR^{p\x n}$, $E\in \mR^{n\x p}$ are constant matrices), and a nonlinear feedback whose output $z$ satisfies
\begin{equation}
\label{yz}
  z \in y + \Theta
\end{equation}
at any time,
where $\Theta  $  is a given convex compact in $\mR^p$. Here, $K+L:= \{u+v: u \in K, v\in L\}$ is the Minkowski sum of sets $K, L\subset \mR^p$ (which is the usual sum of vectors if the sets are singletons). Also, $\dot{(\ )}$ is the time derivative,  and the time arguments of signals will often be omitted for the sake of brevity. The inclusion (\ref{yz}) represents the effect of mechanical backlash
whose simplified (inertialess and frictionless)  model \cite{KP_1989}  is provided by
\begin{equation}
\label{zdot}
    \dot{z}
    =
    P_{N_{y+\Theta  }(z)}(\dot{y})
    =
    P_{N_\Theta  (q)}(\dot{y}),
\end{equation}
where
\begin{equation}
\label{q}
    q:= z-y
\end{equation}
is a $\Theta  $-valued function of time due to (\ref{yz}).
Here, for a closed convex set $S\subset
\mR^r$,
\begin{equation}
\label{N}
    N_S(u)
    :=
    \Big\{
        s\in \mR^r:\
        \inf_{v\in S}
            s^\rT(v-u)
        \> 0
    \Big\}
\end{equation}
 denotes the cone of inward normals to $S$ at a point $u\in \mR^r$, and
\begin{equation}
\label{P}
    P_S(u)
    :=
    \argmin_{v\in S}
    |u-v|
\end{equation}
is the metric projection of $u$ onto $S$. Also, $\mR^r$
(and other Euclidean spaces in consideration)
is endowed with the inner product $u^\rT v$ and the standard Euclidean norm $|u|:= \sqrt{u^\rT u}$ for vectors $u, v\in \mR^r$, where $(\cdot)^\rT$ is the transpose (vectors are organised as columns unless indicated otherwise). The second equality in (\ref{zdot}) follows from the invariance $N_S(u) = N_{S+d}(u+d)$ of the cone (\ref{N}) with respect to the translation of $S$ and $u$ by a common vector $d\in \mR^r$. The projection (\ref{P}) is nonexpanding in the sense that
\begin{equation}
\label{PP}
    |P_S(u) - P_S(v)| \< |u-v|,
    \qquad
    u, v \in \mR^r.
\end{equation}
In (\ref{xy}), the signals $w$, $x$, $y$ are interpreted as the external input, internal state and the output of the linear part of the system; see Fig.~\ref{fig:sys}.
\begin{figure}[htpb]
\centering
\unitlength=1.5mm
\linethickness{0.2pt}
\begin{picture}(65.00,40.00)
    \put(7,25){\makebox(0,0)[cc]{{$w$}}}
    \put(35,12){\makebox(0,0)[cb]{{$z$}}}
    \put(35,38){\makebox(0,0)[ct]{{$y$}}}
    \put(10,25){\vector(1,0){10}}
    \put(20,20){\framebox(10,10)[cc]{{linear}}}
    \put(40,20){\framebox(10,10)[cc]{{backlash}}}
    \put(25,30){\line(0,1){5}}
    \put(25,15){\vector(0,1){5}}
    \put(45,20){\line(0,-1){5}}
    \put(45,15){\line(-1,0){20}}

    \put(45,35){\vector(0,-1){5}}
    \put(45,35){\line(-1,0){20}}
\end{picture}\vskip-17mm
\caption{The closed-loop system (\ref{xy})--(\ref{zdot}) consisting of a linear part and  a nonlinear feedback with a backlash.
}
\label{fig:sys}
\end{figure}
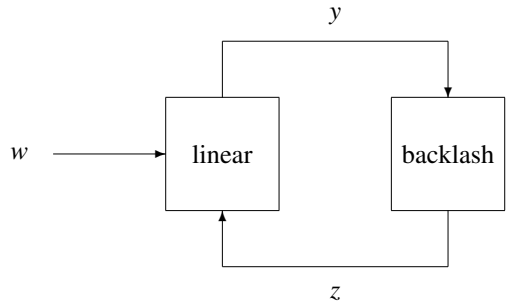
The initial conditions for the system (\ref{xy})--(\ref{zdot}) are specified by $x(0) \in \mR^n$ and $z(0) \in Cx(0) + \Theta  $ (or $q(0) \in \Theta$).
The input $w$ is assumed to be locally integrable, so that
\begin{equation}
\label{wgood}
    \int_0^T |w(t)|\rd t <+\infty
\end{equation}
for any time horizon $T>0$. Accordingly, $x$, $y$, $z$ are absolutely continuous functions of time, which can be verified by using the property
\begin{equation}
\label{zy}
    |\dot{z}|\< |\dot{y}|
\end{equation}
(following from (\ref{zdot}), (\ref{PP})) and an appropriate Gronwall-Bellman lemma estimate.

The backlash dynamics (\ref{zdot}) are a particular case of the Moreau sweeping process \cite{M_1971,M_1977} (see also \cite[pp.~158--160]{KP_1989} and \cite{V_1997}). In fact, as a corollary of the properties of the strong solutions to the Moreau process, the backlash output $z$ is well-defined for any absolutely continuous input $y$. This allows the solution for (\ref{xy})--(\ref{zdot}) to be obtained as the limit of the Picard
iterations (or, alternatively, by using the time discretization and consecutive projections as in the Moreau process), with the resulting solution depending continuously on the initial conditions.

With the backlash input $y$ being a differentiable function of time, the output $z$ remains at rest as long as it is in the interior $\interior (y+\Theta  )$ of the set on the right-hand side of (\ref{yz}) (such moments of time form an open subset of $\mR_+$).  This follows from (\ref{zdot})--(\ref{P}) and the property that the cone (\ref{N}) reduces to  the singleton $N_S(u)=\{0\}$ for any
$u\in \interior S$. The point $z$ can be displaced (along an inward normal, closest to $\dot{y}$) only when it is moved by the boundary $\d(y+\Theta  )$. Until such a contact takes place,
the system behaves like a linear  one.

\section{\bf Stationary initial conditions for the backlash output}
\label{sec:stat}

In order to describe the linear dynamics of the system with a constant backlash output $z$,
 we note that
for a given $T>0$, the set
\begin{equation}
\label{stat}
    \fR(T):= \bigcap_{0\< t\< T} (y(t)+\Theta  )
\end{equation}
(which is also a convex compact in $\mR^p$)
consists of those initial conditions $z(0)$ where the backlash output $z(t)$ remains at rest (that is, $z(t) = z(0)$) for all $t\in [0,T]$. Any point $z(0) \in \fR(T)$ gives rise to a constant forcing term $Ez(0)$ for the ODE in  (\ref{xy}) over the time interval $[0,T]$ and will be referred to as a \emph{$T$-stationary initial condition} for the backlash  output. For the theorem below, we will use an entire function
\begin{equation}
\label{Psi}
    \Psi(u):= \int_0^1 \re^{uv}\rd v
    =
    \left\{
    {\begin{matrix}
    1 & {\rm if}\  u=0\\
    \frac{\re^u-1}{u} & {\rm if}\  u\ne 0
    \end{matrix}}
    \right..
\end{equation}
Its extension from the complex plane to
square matrices \cite{H_2008} describes the solution
$
    \xi(t) := t\Psi(tA) \omega
$ of the ODE $\dot{\xi} = A\xi + \omega$ with a constant forcing term $\omega \in \mR^n$ and zero initial condition $\xi(0) = 0$. If $\det A \ne 0$,  this solution reduces to $\xi(t) = (\re^{tA}-I_n)A^{-1} \omega$, where $I_n$ is the identity matrix of order $n$.

\begin{thm}
\label{th:stat}
For a given time horizon $T>0$, suppose the matrix
\begin{equation}
\label{Phi}
  \Phi(t):= I_n - tC\Psi(tA) E,
\end{equation}
associated with the matrices $A$, $C$, $E$ in (\ref{xy}) and the function $\Psi$ in (\ref{Psi}),
is nonsingular for all $t \in [0,T]$.  Then the set (\ref{stat})  of $T$-stationary initial conditions for the backlash can be represented as
\begin{equation}
\label{stat1}
  \fR(T)
  =
  \bigcap_{0\< t\< T}
  \big(
  \Phi(t)^{-1}
  (
    C
    x_*(t)
    +
    \Theta
  )\big),
\end{equation}
where
\begin{equation}
\label{xz0}
    x_*(t)
    :=
    \re^{tA} x(0) + \int_0^t \re^{(t-s)A}B w(s)\rd s
\end{equation}
is the solution of the ODE in (\ref{xy}) in the absence of the feedback. \hfill$\square$
\end{thm}
\begin{proof}
In view of (\ref{stat}), the inclusion $z(0)\in \fR(T)$ is equivalent to
\begin{equation}
\label{zyZ}
    z(0) \in y(t)+\Theta  ,
    \qquad
    0\< t\< T,
\end{equation}
with the state $x$ satisfying the ODE  (\ref{xy}) driven by the constant backlash output $z(0)$ over the time interval $[0,T]$. Therefore,
\begin{align}
\nonumber
    x(t)
    & =
    \re^{tA}x(0) + \int_0^t \re^{(t-s)A}(Bw(s) + Ez(0))\rd s\\
\label{xz}
    & =
    x_*(t) + t \Psi(tA) Ez(0),
    \qquad
    0\< t\< T,
\end{align}
where $x_*$ is given by (\ref{xz0}).
Substitution of (\ref{xz}) into the second equality in (\ref{xy}) represents  (\ref{zyZ}) as
\begin{align}
\nonumber
    z(0)
    & \in
    C(x_*(t) + t \Psi(tA) Ez(0))+\Theta  \\
\label{zin}
    & = t C\Psi(tA) Ez(0)+Cx_*(t) + \Theta  ,
    \qquad
    0\< t\< T.
\end{align}
In view of (\ref{Phi}),  the inclusion in (\ref{zin}) is equivalent to
$
    \Phi(t)
    z(0)
    \in
    Cx_*(t) + \Theta
$,  that is,
\begin{equation}
\label{z0in}
    z(0)
    \in
    \Phi(t)^{-1}
    (Cx_*(t) + \Theta  ),
    \qquad
    0\< t\< T,
\end{equation}
since it is assumed that $\det \Phi(t)\ne 0$ for all $t\in [0,T]$.
In turn, (\ref{z0in}) is equivalent to
$
    z(0)
    \in
    \bigcap_{0\< t\< T}
    \big(
        \Phi(t)^{-1}
        (Cx_*(t) + \Theta  )
    \big)
$,
which establishes (\ref{stat1}).
\end{proof}

The representation (\ref{stat1}) shows that the set-valued function $\fR$ is nonincreasing: $\fR(t)\subset \fR(s)$ for all $t\> s\> 0$. Therefore, if $\fR(T)=\emptyset$ for some $T>0$, then $\fR(t)=\emptyset$ for all $t\>T$. In this case, the backlash has no $T$-stationary initial conditions and its boundary becomes active over the time interval $[0,T]$ in the sense that
\begin{equation}
\label{PL}
    \int_0^T |\dot{z}(t)|\rd t>0.
\end{equation}
A lower estimate for the path length (\ref{PL})  of the backlash output can be obtained by using a localization of the system trajectories.

\section{\bf Tubular localization of system trajectories}
\label{sec:loc}

If the backlash set $\Theta  $  were a singleton ($\{0\}$ for simplicity) then (\ref{yz}) would lead to $z=y$, and the system  would be governed by a linear ODE
\begin{equation}
\label{xieta}
  \dot{\xi} = A\xi + Bw + E\eta =  F\xi + Bw,
  \qquad
  \eta := C\xi,
\end{equation}
with the dynamics matrix
\begin{equation}
\label{F}
  F:= A + EC.
\end{equation}
The term $EC$ pertains to the feedback which becomes inactive, for example, if at least one of the matrices $C$ or $E$ vanishes.
Although the original system is nonlinear in the case of a nontrivial set $\Theta  $,  the state $x$ admits a localization about the trajectory of its ``linearised'' counterpart specified by  (\ref{xieta}), (\ref{F}) with the same initial condition $\xi(0):=x(0)$.
To this end, use will be made of the bijective correspondence  between convex compacts $S\subset \mR^r$ and their support functions \cite{R_1970}
\begin{equation}
\label{supp}
    \sigma_S(u)
    :=
    \max_{v\in S}
    u^\rT v,
    \qquad
    u \in \mR^r.
\end{equation}
These functions lend themselves to closed-form computation, for example, in the case of ellipsoids. Indeed, for an ellipsoid $S:= \{v \in \mR^r:\ \|v-c\|_{\Sigma^{-1}}\< 1\}$ with centre $c\in \mR^r$ and ``matrix  radius'' $\sqrt{\Sigma}$ (where $\Sigma=\Sigma^\rT \in \mR^{r\x r}$ is a positive definite matrix), the support function (\ref{supp}) takes the form $\sigma_S(u) = \max_{v\in \mR^r:\, |v|\< 1} u^\rT (c + \sqrt{\Sigma} v)= u^\rT c + \|u\|_\Sigma$, where $ \|u\|_\Sigma := |\sqrt{\Sigma} u| = \sqrt{u^\rT \Sigma u}$.

\begin{lem}
\label{lem:supp}
Suppose the input $w$ of the system (\ref{xy})--(\ref{zdot}) is locally integrable. Then, at any time $t\>0$, the system state satisfies
\begin{equation}
\label{xX}
    x(t) \in
    \xi(t) + \Xi (t),
\end{equation}
where
\begin{equation}
\label{xi}
  \xi(t):= \re^{tF} x(0) + \int_0^t \re^{(t-s)F} Bw(s)\rd s
\end{equation}
is the solution of the ODE in (\ref{xieta}) with the initial condition $\xi(0):=x(0)$. Here,   $\Xi (t)$ is a time-varying convex compact in $\mR^n$ whose support function (\ref{supp}) is linearly related to that of the backlash set $\Theta  $:
\begin{equation}
\label{Xsupp}
  \sigma_{\Xi (t)}(u) = \int_0^t \sigma_\Theta  (E^\rT \re^{sF^\rT}u)\rd s,
  \qquad
  u \in \mR^n.
\end{equation}
\hfill$\square$
\end{lem}
\begin{proof}
In view of (\ref{F}),  the ODE in (\ref{xy}) can be represented as
\begin{equation}
\label{xy1}
  \dot{x} = Ax + B w + E (y +q) = Fx + Bw + Eq,
\end{equation}
where $q$ is given by (\ref{q}). The solution of (\ref{xy1}) satisfies
\begin{equation}
\label{xsol}
  x(t) - \xi(t) =  \int_0^t \re^{sF}Eq(t-s)\rd s
\end{equation}
in view of (\ref{xi}) and the linear superposition principle.
Now, for any $u \in \mR^n$,
\begin{equation}
\label{sup}
    \sup_{q:[0,t]\to \Theta  }
    u^\rT
    \int_0^t \re^{sF} Eq(t-s)\rd s
    \<
    \int_0^t
    \max_{v \in \Theta  }
    (
    u^\rT
    \re^{sF} Ev
    )
    \rd s
    =
    \int_0^t
    \sigma_\Theta  (E^\rT \re^{sF^\rT}u)
    \rd s,
\end{equation}
where the supremum is over measurable $\Theta  $-valued functions on the time interval $[0,t]$.\footnote{Note that the inequality in (\ref{sup}) holds as an equality (although this does not affect the results).} The right-hand side of (\ref{sup}) is the support function of the convex compact $\Xi (t) \subset \mR^n$ from (\ref{Xsupp}), and hence, the vector on the right-hand side of (\ref{xsol}) belongs to $\Xi (t)$, thus establishing (\ref{xX}).
\end{proof}

The above proof shows that Lemma~\ref{lem:supp} employs only the inclusion (\ref{yz}) and is valid regardless of the specific backlash dynamics (\ref{zdot}). For simplicity, it is assumed   in what follows that
\begin{equation}
\label{0inZ}
    0\in \Theta  ,
\end{equation}
and hence, $\sigma_\Theta  (u)\>0$ for all $u\in \mR^p$.  In this case, the right-hand side of (\ref{Xsupp}) is nondecreasing in time $t$, and so also is the set-valued map $\Xi $ (that is, $\Xi (t)\supset \Xi (s)$ for all $t\> s\> 0$). This leads to a bounded limit
\begin{equation}
\label{Xinf}
    \Xi _\infty:= \clos \bigcup_{t\>0} \Xi (t)
\end{equation}
(with $\clos(\cdot)$ the closure of a set), provided the matrix $F$ in (\ref{F}) is Hurwitz. In view of (\ref{Xsupp}), the support function of the limit set (\ref{Xinf}) is
\begin{equation}
\label{Xinfsupp}
  \sigma_{\Xi _\infty}(u)
  =
  \int_0^{+\infty} \sigma_\Theta  (E^\rT \re^{sF^\rT}u)\rd s,
  \qquad
  u \in \mR^n.
\end{equation}
Note that the convergence in (\ref{Xinf}) is exponentially fast. More precisely,  (\ref{Xsupp}), (\ref{Xinfsupp}) imply that
\begin{align}
\nonumber
    D(\Xi _\infty, \Xi (t))
     &= \max_{u\in \mS_n}|\sigma_{\Xi _\infty}(u)-\sigma_{\Xi (t)}(u)|\\
\nonumber
    & =
    \max_{u\in \mS_n}
    \int_t^{+\infty} \sigma_\Theta  (E^\rT \re^{sF^\rT}u)\rd s\\
\nonumber
    & \<
    \int_t^{+\infty}
        \max_{u\in \mS_n}
    \sigma_\Theta  (E^\rT \re^{sF^\rT}u)\rd s\\
\label{DXX}
    & \<
    D(\Theta  ,\{0\})
    \|E\|
    \int_t^{+\infty}
    \|\re^{sF}\|
    \rd s,
\end{align}
where
$
    \mS_n:= \{u \in \mR^n:\ |u|=1\}
$
is the unit sphere in $\mR^n$.
Here, use is made of the Hausdorff deviation
\begin{equation}
\label{DLM}
    D(L,M)
    :=
    \sup_{u\in L}
    \rho(u,M)
\end{equation}
of a set $L\subset \mR^r$ from another set $M\subset \mR^r$,
with
\begin{equation}
\label{rho}
    \rho(u,M)
    :=
    \inf_{v\in M}
    |u-v|
\end{equation}
denoting the distance from a point $u\in \mR^r$ to $M$. In view of (\ref{DLM}), (\ref{rho}), (\ref{supp}),
\begin{equation}
\label{DTheta0}
    D(\Theta,\{0\}) = \max_{v \in \Theta  }|v| = \max_{u\in \mS_n} \sigma_\Theta  (u)
\end{equation}
quantifies the deviation of the backlash set $\Theta  $ from the origin.
Also, $\|\cdot\|$ in (\ref{DXX}) is the operator norm of matrices (induced by the Euclidean vector norm $|\cdot|$), and its submultiplicativity is used. From (\ref{DXX}), it follows that
$    \limsup_{t\to +\infty}
    \big(
    \frac{1}{t}
    \ln
    D(\Xi _\infty, \Xi (t))
    \big)
    \<
    \mu$,
where
\begin{equation}
\label{stab}
  \mu
  :=
  \max_{1\< k\< n} \Re \lambda_k
  =
  \ln \br(\re^F)
  <0
\end{equation}
is the largest real part of the eigenvalues $\lambda_1, \ldots, \lambda_n$
of the Hurwitz matrix $F$ in (\ref{F}), so that $\frac{1}{|\mu|}$ quantifies the decay time for transient processes in the linearised system (\ref{xieta}), with $\br(\cdot)$ the spectral radius of a square matrix.  Therefore,  on time scales $t\gg \frac{1}{|\mu|}$, the inclusion
\begin{equation}
\label{xX1}
    x(t) \in \xi(t) + \Xi _\infty
\end{equation}
(obtained by replacing the set $\Xi (t)$ in (\ref{xX}) with its limit $\Xi _\infty$ from (\ref{Xinf})) is only slightly more conservative than (\ref{xX}). The right-hand side of (\ref{xX1}) can be viewed as a ``tube'' about the  trajectory $\xi$ of the linearised system in (\ref{xi}), with its ``cross section'' being specified by the set $\Xi _\infty$. Note that $\Xi _\infty$ does not depend on the input $w$ which enters the right-hand side of (\ref{xX1}) only through  $\xi$.

\section{\bf Backlash boundary  activation for periodic inputs}
\label{sec:act}

Consider the system dynamics when the external input $w$ is a $T$-periodic bounded function of time. More precisely, suppose $w(t+T)=w(t)$ for all $t\>0$, and $w|_{[0,T]}$ belongs to the Banach space $ L_\infty([0,T], \mR^m)$
with the norm $\|w\|_\infty:= \esssup_{0\< t\< T} |w(t)|$, whereby the local integrability condition (\ref{wgood}) is also satisfied.
With the matrix $F$ in (\ref{F}) being assumed to be Hurwitz, the ODE in (\ref{xieta}) has a unique $T$-periodic solution $\xi_T$ with the initial condition
\begin{equation}
\label{xT}
  \xi_T(0) = (I_n-\re^{TF})^{-1} \int_0^T \re^{(T-t)F} Bw(s)\rd s.
\end{equation}
Substitution of (\ref{xT}) into (\ref{xi}) represents this solution in the form
\begin{align}
\nonumber
  \xi_T(t)
  & =
  \re^{tF}
  (I_n-\re^{TF})^{-1}
  \int_0^T \re^{(T-s)F} Bw(s)\rd s
  +
  \int_0^t \re^{(t-s)F} Bw(s)\rd s\\
\nonumber
  & =
  ((I_n-\re^{TF})^{-1}-I_n)
  \int_0^T \re^{(t-s)F} Bw(s)\rd s
  +
  \int_0^t \re^{(t-s)F} Bw(s)\rd s\\
\label{xiT}
  & =
  \int_0^T
  ((I_n-\re^{TF})^{-1} - \chi_{[t,T]}(s) I_n)
  \re^{(t-s)F} Bw(s)\rd s
\end{align}
for all $t \in [0,T]$,
where $\chi_S(\cdot)$ is the indicator function of a set $S$. Since the linearised system (\ref{xieta}) is stable and  the external input $w$ is $T$-periodic, then
$\xi_T$ in (\ref{xiT}) is a $T$-periodic global attractor for the system state $\xi$ in the sense that \begin{equation}
\label{xixi}
    \lim_{t\to +\infty} |\xi(t)-\xi_T(t)| = 0
\end{equation}
for any initial condition $\xi(0)$.
A similar property holds for the corresponding $T$-periodic output
\begin{equation}
\label{etaT}
  \eta_T := C \xi_T
\end{equation}
with respect to the output $\eta$ of the linear system (\ref{xieta}).
In view of (\ref{xX1}), all those initial conditions $x(0)$, $z(0)$, which give rise to $T$-periodic  trajectories of the nonlinear system  (\ref{xy})--(\ref{zdot}), belong to the convex compact
\begin{equation}
\label{loc}
    x(0) \in \xi_T(0)+\Xi_\infty,
    \
    z(0) \in C x(0) + \Theta \subset \eta_T(0)+C\Xi_\infty+\Theta.
\end{equation}
The Poincare map, associated with the nonlinear system, is continuous and, by the Brouwer fixed-point  theorem \cite{AG_1980},  has at least one fixed point in the set (\ref{loc}). We will  be concerned with conditions which secure  uniqueness  for the forced periodic regime in the nonlinear system and exponentially fast convergence to it. To this end,  the following theorem obtains asymptotic bounds  for the backlash output path length (see also \cite{V_1997} for a similar bound) and its time derivative
by using the tubular localization from Section~\ref{sec:loc}.
For its formulation, we denote by
\begin{equation}
\label{diam}
    \diam (S):= \sup_{u,v\in S}|u-v|
\end{equation}
the diameter of a bounded set $S \subset \mR^r$. Accordingly, the oscillation of a vector-valued function $f$ on a set $ K$ is the diameter
\begin{equation}
\label{osc}
    \Omega_K(f)
    :=
    \sup_{s,t\in K}
    |f(s)-f(t)|
    =
    \diam(f(K))
\end{equation}
of the image $f(K):= \{f(t): t \in K\}$ of $K$ under the map $f$.

\begin{thm}
\label{th:path}
Suppose the matrix $F$ of the linearised system in (\ref{xieta}), (\ref{F}) is Hurwitz, and the backlash set $\Theta  $ satisfies (\ref{0inZ}). Also, let the nonlinear system (\ref{xy})--(\ref{zdot}) be driven by a $T$-periodic bounded input $w$. Then, for any initial condition of the system,  the path length for the backlash output $z$ satisfies
\begin{equation}
\label{path}
    \liminf_{\tau\to +\infty}
    \int_\tau^{\tau+T} |\dot{z}(t)|\rd t
    \>
    (\mho -    d)_+
\end{equation}
(with $(\cdot)_+ := \max(\cdot, 0)$  the positive cutoff function), where
\begin{equation}
\label{mho}
  \mho:= \Omega_{[0,T]}(\eta_T)
\end{equation}
is  the oscillation (\ref{osc}) of the $T$-periodic output (\ref{etaT}) of the linearised system, and
\begin{equation}
\label{d}
      d:= \diam(C\Xi _\infty + \Theta  ),
\end{equation}
is the diameter (\ref{diam}) of the set $C\Xi _\infty + \Theta$ associated with $\Xi _{\infty}$ in (\ref{Xinf}), (\ref{Xinfsupp}). Furthermore,
\begin{equation}
\label{zinf}
  \limsup_{t\to +\infty}
  |\dot{z}(t)|
  \<
  \|\dot{\eta}_T\|_\infty
  +
  D(C(F\Xi_\infty + E\Theta), \{0\}).
\end{equation}
\hfill$\square$
\end{thm}
\begin{proof}
From (\ref{q}) and  the second equalities in (\ref{xy}), (\ref{xieta}), it follows that
\begin{equation}
\label{diffs}
    z-\eta = C(x-\xi) + q
    \in C\Xi _\infty + \Theta,
\end{equation}
where use is made of (\ref{xX1}) under the conditions (\ref{0inZ}) and the matrix $F$ in (\ref{F}) being Hurwitz.
A combination of the triangle inequality  with (\ref{diffs}), (\ref{d}) leads to
\begin{align}
\nonumber
    |\eta(s)-\eta(t) |
     &=
    |z(s)-z(t)+  z(t)-\eta(t)-(z(s)-\eta(s))|    \\
\nonumber
    & \<
    |z(s)-z(t)|
    +
    |z(t)-\eta(t)-(z(s)-\eta(s))|    \\
\label{diffs1}
    & \<
    |z(s)-z(t)| + d
\end{align}
for all $s,t\>0$. By taking the supremum on both sides of (\ref{diffs1}) with respect to $s,t \in K$ over a time interval $K$, it follows that the corresponding oscillations (\ref{osc}) of the functions $\eta$, $z$  satisfy
\begin{equation}
\label{oscd}
\Omega_K(\eta)
   \<
    \Omega_K(z)
    +
  d \<
  \int_K |\dot{z}(t)|\rd t + d,
\end{equation}
so that
\begin{equation}
\label{diff2}
      \int_K |\dot{z}(t)|\rd t
      \>
      (\Omega_K(\eta)
      -
      d)_+.
\end{equation}
Since $\eta_T$ in (\ref{etaT})  is a $T$-periodic attractor for the  output $\eta$ of the stable linear system (\ref{xieta}) with the $T$-periodic external input $w$ (in the sense that $\lim_{t\to +\infty} |\eta(t)-\eta_T(t)| = 0$), then
\begin{equation}
\label{diff3}
  \lim_{\tau\to +\infty}
  \Omega_{[\tau, \tau+T]}(\eta)
  =
  \Omega_{[0, T]}(\eta_T).
\end{equation}
The inequality (\ref{path}) is now obtained by combining (\ref{diff2}) with (\ref{diff3}) and using (\ref{mho}). In order to prove (\ref{zinf}), we note that
\begin{align}
\nonumber
    \dot{y}
    & = C\dot{x} = C(Fx + Bw + Eq)\\
\nonumber
     & = C(F\xi_T+ Bw+ F(\xi-\xi_T) + F(x-\xi) +  Eq)\\
\nonumber
    & = C\dot{\xi}_T + C(F(\xi-\xi_T) + F(x-\xi) +  Eq)    \\
\label{ydot}
    & \in \dot{\eta}_T + CF(\xi-\xi_T) + C(F\Xi_\infty +  E\Theta),
\end{align}
where use is made of (\ref{xy1}), (\ref{xX1}), (\ref{xieta}). Since $\dot{\eta}_T$ is  a $T$-periodic bounded function of time, a combination of (\ref{ydot}) with (\ref{xixi}) leads to
\begin{equation}
\label{yinf}
  \limsup_{t\to +\infty}
  |\dot{y}(t)|
  \<
  \|\dot{\eta}_T\|_\infty
  +
  D(C(F\Xi_\infty + E\Theta), \{0\}),
\end{equation}
where the last term is similar to (\ref{DTheta0}). The relation (\ref{zinf}) follows from (\ref{yinf}) in view of (\ref{zy}).
\end{proof}

The lower bound (\ref{path}) provided by Theorem~\ref{th:path} pertains to a general problem of finding the minimum arc length of a curve in a tubular neighbourhood of another curve.
A similar argument, using (\ref{diffs1}), (\ref{oscd}), leads to
\begin{equation}
\label{path1}
    \liminf_{\tau\to +\infty}
    \int_\tau^{\tau+T} |\dot{z}(t)|\rd t
    \>
    \sup
    \sum_{k=1}^N
    (|\eta_T(t_k)-\eta_T(t_{k-1})| - d)_+,
\end{equation}
where the supremum is taken over all partitions  $0=t_0 < t_1 < \ldots < t_N=T$ of the interval $[0,T]$ into $N=1,2,3,\ldots$ subintervals. In particular, if the points $\eta_T(t_1), \ldots, \eta_T(t_N)$ are centres of pairwise disjoint open balls of radius $\eps >0$, thus giving rise to an $\eps$-packing of the set $\eta_T([0,T]) \subset \mR^p$, then the corresponding sum on the right-hand side of (\ref{path1}) satisfies
$
    \sum_{k=1}^N
    (|\eta_T(t_k)-\eta_T(t_{k-1})| - d)_+ \> (2\eps-d)_+ N$. Therefore,
\begin{equation}
\label{path2}
    \liminf_{\tau\to +\infty}
    \int_\tau^{\tau+T} |\dot{z}(t)|\rd t
    \>
    \sup_{\eps > d/2 }
    ((2\eps-d)N_\eps ),
\end{equation}
where $N_\eps$ is the largest cardinality of an $\eps$-packing of the set $\eta_T([0,T]) $ (so that $\log_2 N_\eps$ is the $\eps$-capacity \cite{K_1956} of this set). Note that the right-hand side of (\ref{path2}) is amenable to asymptotic analysis as $d\to 0+$.

Now, the right-hand side of (\ref{path}) is positive if the external input $w$ has a sufficiently large ``amplitude'' in the sense that the oscillation $\mho$ in (\ref{mho})
exceeds the quantity $d$ in (\ref{d}) which does  not depend on $w$; see Fig.~\ref{fig:loop}.
\begin{figure}[htb]
\begin{center}
\includegraphics[scale=0.3]{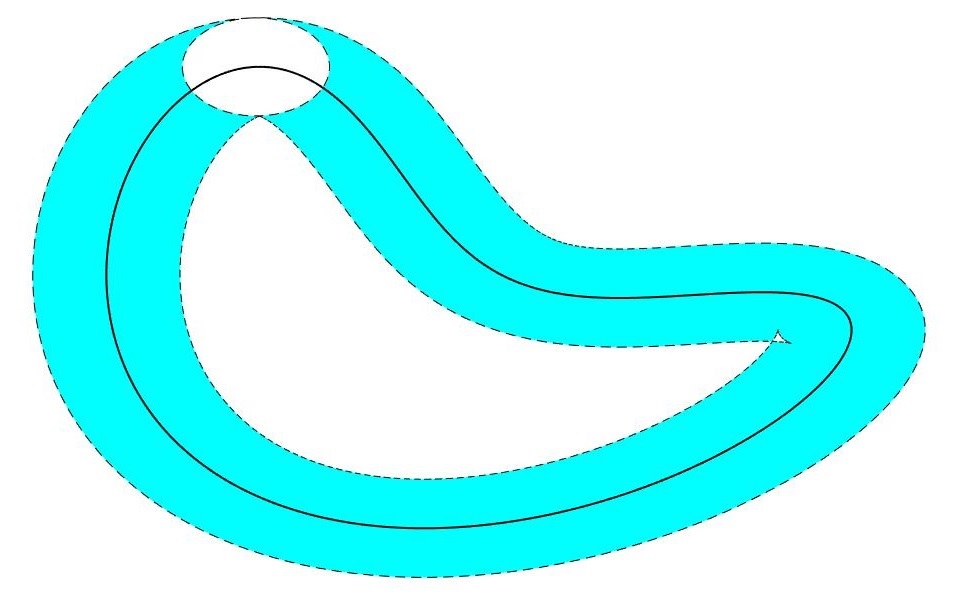}
\caption{An illustration of the tubular localization (between the dashed lines) for periodic trajectories $z$ of the original nonlinear system about the periodic trajectory $\eta_T$ (solid curve) for the linearised system. The ellipse represents the set $C\Xi _\infty + \Theta$ in (\ref{d}) associated with  $\Xi _\infty$ from (\ref{Xinf}), (\ref{Xinfsupp}).}
\label{fig:loop}
\end{center}
\end{figure}
 In this case,
the backlash output $z$ has no stationary initial conditions (see Section~\ref{sec:stat}) and  is eventually forced to move (its path length over any sufficiently distant time interval of duration $T$ is separated from zero). This backlash motion activates dissipation inequalities for the system under additional geometric constraints.

\section{\bf Dissipation relations using strong convexity}
\label{sec:diss}

For what follows, suppose the backlash set $\Theta  $ in (\ref{yz}) is a strongly convex subset of $\mR^p$ in the sense that the quantity
\begin{equation}\label{R}
    R
    :=
    \inf
    \Big\{
        r\> 0:\,
        \Theta
        \subset
        \bigcap_{
            u\in \d \Theta  ,\,
            v \in N_\Theta  (u) \bigcap \mS_p
        }
        \overline{B}_r(u+rv)
    \Big\}
\end{equation}
is finite, where $\overline{B}_r(c):= \clos B_r(c)$ is the closed ball with
centre $c$ and radius $r$.  In this case, $R$ is referred to as the \emph{strong convexity constant} of $\Theta  $.
Therefore, $R$ is the smallest $r$ such that for any supporting
hyperplane for the set $\Theta  $ at any given boundary point $u \in
\d \Theta  $ with a unit inward normal $v \in N_\Theta  (u)\bigcap \mS_p$, the closed ball
$\overline{B}_r(u+rv)$, supported at $u$ by the same hyperplane,
contains $\Theta  $.
This definition of strong convexity is essentially equivalent
to those in \cite{FO_1981,P_1975} and \cite[pp.~164--165]{KP_1989}.
Alternatively, the quantity $R$ in (\ref{R}) is
the smallest $r$ for which $\Theta  $ can be represented as an intersection
of closed balls of radius $r$.

Consider two trajectories of the underlying nonlinear system (\ref{xy})--(\ref{zdot}) driven by the same input $w$, but with different initial conditions $x_k(0) \in \mR^n$ and $z_k(0)\in y_k(0)+\Theta  $, $k=1,2$. The deviations of signals, associated with these trajectories, are given by
\begin{align}
\label{X_Y}
    X
    & := x_1 - x_2,
    \qquad
    Y
     := y_1 - y_2 = CX,\\
\label{Z_Q}
    Z
    & := z_1 - z_2,
    \qquad
    Q
     := q_1 - q_2 = Z - Y
\end{align}
and satisfy the ODEs
\begin{align}
\label{Xdot}
  \dot{X}
  & = AX + EZ
  =
  FX + EQ,\\
\label{Zdot}
    \dot{Z}
    & =
    P_{N_\Theta  (q_1)}(\dot{y}_1)
    -
    P_{N_\Theta  (q_2)}(\dot{y}_2)
\end{align}
in accordance with (\ref{xy1}), (\ref{F}). Due to the strong convexity of the backlash set $\Theta  $ with the constant $R$,
application of (\ref{lem1:2}) of Lemma~\ref{lem:diss}
to (\ref{Zdot})
leads to
\begin{equation}
\label{QZdot}
        Q^\rT
        \dot{Z}
    \<
    -
    \gamma
    |Q|^2,
\end{equation}
where
\begin{equation}
\label{zz}
    \gamma
    :=
    \frac{1}{2R}
    (|\dot{z}_1|+|\dot{z}_2|)
\end{equation}
is a locally integrable nonnegative function of time.\footnote{the case of the usual (rather than strong) convexity is obtained formally  by letting $R\to +\infty$, in which case, $\gamma=0$  and (\ref{QZdot}) reduces to $        Q^\rT
        \dot{Z} \< 0
$.}
 By substituting $Z = Q + Y$ from   (\ref{Z_Q}) into (\ref{QZdot}) and using the identity $Q^\rT \dot{Q} = \frac{1}{2}(|Q|^2)^{^\centerdot}$,
 it follows that
\begin{equation}
\label{Q2dot}
    (|Q|^2)^{^\centerdot}
    +
    2Q^\rT
        \dot{Y}
    \<
    -
    2\gamma
    |Q|^2.
\end{equation}
A combination of (\ref{X_Y}) with (\ref{Xdot}) allows (\ref{Q2dot}) to be represented as
\begin{equation}
\label{Q2dot1}
    (|Q|^2)^{^\centerdot}
    \<
    -
    2X^\rT F^\rT C^\rT Q
    -
    2Q^\rT (\bS(CE) + \gamma I_p) Q,
\end{equation}
where
$
    \bS(M):= \frac{1}{2}(M+M^{\rT})
$
is the symmetrizer of square matrices.
In order to quantify the deviation of the system trajectories, consider the following candidate for a Lyapunov function:
\begin{equation}
\label{V_Gamma}
    V
    :=
    \|X\|_{\Pi}^2
    +
    |Q|^2
    =
    \left\|
    \begin{bmatrix}
      X\\
      Q
    \end{bmatrix}
    \right\|_{\Gamma}^2,
    \qquad
    \Gamma
    :=
    \begin{bmatrix}
      \Pi & 0\\
      0 & I_p
    \end{bmatrix},
\end{equation}
which is specified by a real positive definite  symmetric matrix $\Pi$  of order $n$.
In view of (\ref{Xdot}), (\ref{Q2dot1}), the time derivative of $V$ in (\ref{V_Gamma}) satisfies
\begin{align}
\nonumber
    \dot{V}
     & =
    2X^\rT \Pi \dot{X}
    +
    (|Q|^2)^{^\centerdot}\\
\nonumber
    & =
        2X^\rT \bS(\Pi F)X
    +
    2X^\rT \Pi E Q
    +
    (|Q|^2)^{^\centerdot}\\
\nonumber
    & \<
        2X^\rT \bS(\Pi F)X
    +
    2X^\rT \Pi E Q
    -
    2X^\rT F^\rT C^\rT Q
        -
    2Q^\rT (\bS(CE) + \gamma I_p) Q
    \\
\nonumber
    & =
        2X^\rT \bS(\Pi F)X
    +
    4X^\rT HQ
    -
    2Q^\rT (\bS(CE)+\gamma I_p) Q\\
\label{Vdot}
    & =
    2
    \begin{bmatrix}
      X^\rT &
      Q^\rT
    \end{bmatrix}
    G
        \begin{bmatrix}
      X\\
      Q
    \end{bmatrix},
\end{align}
where
\begin{equation}
\label{G_H}
    G
    :=
    \begin{bmatrix}
      \bS(\Pi F) & H\\
      H^\rT & -\bS(CE)-\gamma I_p
    \end{bmatrix},
    \qquad
    H
     := \frac{1}{2}(\Pi E-F^\rT C^\rT)
\end{equation}
(note that the matrix $H$ depends on $\Pi$).
Let $\lambda$ be a fixed but otherwise arbitrary scalar such that
\begin{equation}
\label{lammu}
    0<\lambda < |\mu|,
\end{equation}
with $\mu$ given by (\ref{stab}). Then  there exists a matrix $\Pi\succ 0$ satisfying the algebraic Lyapunov inequality
$
    \Pi F + F^\rT \Pi  \preccurlyeq -2\lambda \Pi
$, whose equivalent form is
\begin{equation}
\label{ALI}
    \bS(\Pi F) \preccurlyeq -\lambda \Pi .
\end{equation}
Since the matrix $\Gamma$ in (\ref{V_Gamma}) is positive definite, (\ref{G_H}), (\ref{ALI}) imply that
\begin{align}
\nonumber
    \Lambda
    & :=
    \Gamma^{-1/2} G \Gamma^{-1/2}\\
\nonumber
    & =
    {\begin{bmatrix}
      \Pi^{-1/2}\bS(\Pi F)\Pi^{-1/2} & \Pi^{-1/2}H\\
      H^\rT\Pi^{-1/2} & -\bS(CE)- \gamma I_p
    \end{bmatrix}}\\
\label{UpsGamma}
    & \preccurlyeq
    {\begin{bmatrix}
      -\lambda I_n  & \Pi^{-1/2}H\\
      H^\rT\Pi^{-1/2} & -(\alpha + \gamma) I_p
    \end{bmatrix}},
\end{align}
where
\begin{equation}
\label{alpha}
    \alpha
    :=
    \lambda_{\min}(\bS(CE))
\end{equation}
depends only on the coupling between the linear part of the system and the backlash in (\ref{xy}).
From the dependence of the matrix $\Lambda$ on $\Gamma$, $G$, it follows that
\begin{equation}
\label{Upsup}
    G
    =
    \sqrt{\Gamma} \Lambda \sqrt{\Gamma}
    \preccurlyeq
    \phi \Gamma,
    \qquad
    \phi:= \lambda_{\max}(\Lambda),
\end{equation}
where $\lambda_{\max}(\cdot)$ is the largest eigenvalue of a real symmetric matrix. Note that
$\phi$ is the largest generalised eigenvalue of the matrix pencil associated with the pair $(G, \Gamma)$ and depends on time through the function $\gamma$ in (\ref{UpsGamma}).
By using Lemma~\ref{lem:spec},
it follows from (\ref{UpsGamma})--(\ref{Upsup}) that
\begin{align}
\nonumber
    \phi
    & \<
    \lambda_{\max}
    \left(
        {\begin{bmatrix}
      -\lambda I_n  & \Pi^{-1/2}H\\
      H^\rT\Pi^{-1/2} & -(\alpha + \gamma) I_p
    \end{bmatrix}}
    \right)\\
\label{phipsi}
    & =
    \sqrt{\beta + \frac{1}{4}(\alpha-\lambda + \gamma)^2 }
    -\frac{1}{2}(\alpha +\lambda +  \gamma)
    =:\psi(\gamma),
\end{align}
where
\begin{equation}
\label{bet}
    \beta:= \|\Pi^{-1/2} H\|^2 = \lambda_{\max}(H^\rT \Pi^{-1} H).
\end{equation}
The function $\psi: \mR_+\to \mR$, defined by (\ref{phipsi}), is a nonincreasing convex function which, in view of (\ref{lammu}), satisfies
\begin{equation}
\label{psilim}
  \psi(0)
     =
    \sqrt{\beta + \frac{1}{4}(\alpha-\lambda)^2 }
    -\frac{\alpha +\lambda}{2}
    \>
    \lim_{u\to +\infty} \psi(u) = -\lambda<0.
\end{equation}
Moreover, if $\beta>0$ in (\ref{bet}), then  $\psi$ is strictly decreasing, strictly convex and continuously differentiable.
In view of (\ref{V_Gamma}), (\ref{Upsup}), (\ref{phipsi}), the right-hand side of (\ref{Vdot}) admits an upper bound
\begin{equation}
\label{GB1}
    \dot{V}
    \<
    2\phi
    \begin{bmatrix}
      X^\rT &
      Q^\rT
    \end{bmatrix}
    \Gamma
    \begin{bmatrix}
      X\\
      Q
    \end{bmatrix}
    =
    2\phi V
    \< 2\psi(\gamma) V.
\end{equation}
Application of the Gronwall-Bellman lemma to (\ref{GB1}) yields
\begin{equation}
\label{VV}
    \sqrt{V(t)} \< \sqrt{V(0)}\, \re^{\int_0^t \psi(\gamma(s))\rd s},
    \qquad
    t\> 0.
\end{equation}
Now, if $\psi(0)<0$, which is equivalent to $\alpha > \frac{\beta}{\lambda}$  in (\ref{psilim}), then  (\ref{VV}) leads to $\sqrt{V(t)} \< \sqrt{V(0)}\, \re^{\psi(0) t}$, thus securing an exponentially fast decay for the Lyapunov function (\ref{V_Gamma}) as $t\to +\infty$ regardless of the strong convexity of the backlash set. In this case, the nonlinear system has a unique forced $T$-periodic regime to which all the system trajectories converge at an exponential rate.  However, in the case when $\psi(0)\>0$, the upper bound (\ref{VV}) can lead to an exponential decay for $V$ only due to the strong convexity of the backlash set and effective movement of the backlash output captured in the function $\gamma$ in (\ref{zz}). In combination with (\ref{zz}) and the Jensen inequality, the convexity of the function $\psi$ in (\ref{phipsi})  implies that
\begin{equation}
\label{Jen}
    \psi(\gamma)
    \<
    \frac{1}{2}
    \big(
    \psi
    (\frac{1}{R}|\dot{z}_1|)
    +
    \psi
    (\frac{1}{R}|\dot{z}_2|)
    \big).
\end{equation}
Now, let
\begin{align}
\label{g1}
    \gamma_1
    & :=
    \frac{1}{TR}(\mho-d)_+,\\
\label{ginf}
    \gamma_\infty
    & :=
    \frac{1}{R}
    (
      \|\dot{\eta}_T\|_\infty
  +
  D(C(F\Xi_\infty + E\Theta), \{0\}))
\end{align}
be associated with the bounds (\ref{path}), (\ref{zinf}) of Theorem~\ref{th:path}, and hence, $\gamma_1 \< \gamma_\infty$. Then (\ref{Jen})--(\ref{ginf}) and a local linear upper bound
\begin{equation}
\label{psiup}
  \psi(u) \< \psi(0) - \frac{u}{\gamma_\infty} (\psi(0)-\psi(\gamma_\infty)),
  \qquad
  0\< u \< \gamma_\infty,
\end{equation}
for the nonincreasing convex function $\psi$ over the interval $[0,\gamma_\infty]$ (see Fig.~\ref{fig:psi1})
\begin{figure}[htb]
\begin{center}
\includegraphics[scale=0.5]{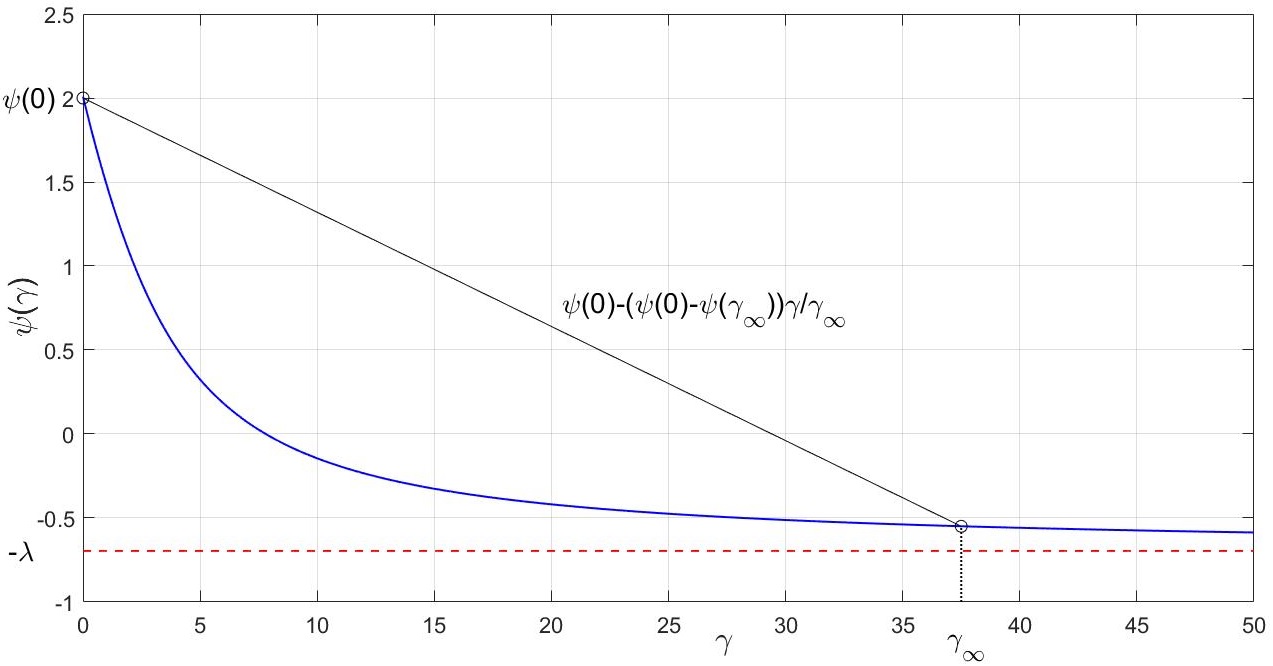}
\caption{A typical graph of the function $\psi$ in (\ref{phipsi}) along with its extreme values (\ref{psilim}) and the linear upper bound (\ref{psiup}) over the interval $[0,\gamma_\infty]$ (in this example, $\psi(0)>0$).}
\label{fig:psi1}
\end{center}
\end{figure}
imply that
\begin{equation}
\label{theta}
    \frac{1}{T}
    \limsup_{\tau\to +\infty}
    \int_{\tau}^{\tau+T}
    \psi(\gamma(t))
    \rd t
    \<
    \psi(0)
    -
    \frac{\gamma_1}{\gamma_\infty}(\psi(0)-\psi(\gamma_\infty)) =:\theta.
\end{equation}
If $\theta<0$, then, in view of (\ref{VV}), it provides an upper bound for the Lyapunov exponent
$$
    \frac{1}{2}
    \limsup_{t\to +\infty}
    \Big(
    \frac{1}{t}
    \ln
    V(t)
    \Big)
    \<
    \limsup_{t\to +\infty}
    \Big(
    \frac{1}{t}
    \int_0^t
    \psi(\gamma(s))
    \rd s
    \Big)
    \<
    \theta,
$$
which quantifies the rate of convergence of the system trajectories with different initial conditions to the unique $T$-periodic trajectory.

Note that $\theta$ in (\ref{theta}) is a convex combination of $\psi(0)$, $\psi(\gamma_\infty)$ and, in view of (\ref{psilim}), can be made negative by an appropriate choice of the $T$-periodic external input $w$. More precisely, this is achieved by increasing $\|\dot{\eta}_T\|_\infty$ in (\ref{ginf}) (which also quantifies the ``amplitude'' of $w$), thus making $\gamma_\infty$ sufficiently large for $\psi(\gamma_\infty)$ to be negative (and, in principle, arbitrarily close to $-\lambda$), and shaping the input $w$ so as to satisfy
$
    \frac{\gamma_1}{\gamma_\infty}
    >
    \frac{\psi(0)}{\psi(0)-\psi(\gamma_\infty)}
$.

The above analysis has employed a local linear upper bound (\ref{psiup}) for the function $\psi$. However, in the case $\beta>0$ in (\ref{bet}), the function $\psi$ is strongly convex over the interval and   admits a subtler quadratic upper bound (see Fig.~\ref{fig:psi1}). Such a refined bound can be combined in (\ref{theta})  with estimates for the second-order moments $\limsup_{\tau\to +\infty}
    \int_{\tau}^{\tau+T}
    |\dot{z}|^2
    \rd t$ of the backlash output.

\section{\bf Conclusion}
\label{sec:conc}

We have considered a class of nonlinear systems, where a linear plant forms a feedback loop with a backlash described by a strongly convex compact. We have discussed a tubular localization  for the trajectories of the system about those of its linearised counterpart without the backlash. This property has been combined with dissipation relations in order to study the establishment of periodic regimes in the system subject to periodic inputs of relatively large amplitudes.
This approach, which is based on enhanced  differential inequalities and was used previously for the Moreau sweeping process  in \cite{V_1997}, is applicable to similar differential inclusions with strongly convex sets. Other directions of research may include extension of these results to system  interconnections with several strongly convex backlashes and also backlash models which take into account the friction and inertial effects.

\appendix
\section*{\bf A. Inequalities  for inward normals}
\label{sec:A}
\renewcommand{\theequation}{A\arabic{equation}}
\setcounter{equation}{0}

For completeness, the following lemma  compares two inequalities for the inward normals to convex and strongly convex sets.

\begin{lem}
\label{lem:diss}
Let $S$ be a closed convex subset of $\mR^p$. Then for any given points
$x,y\in S$,
\begin{equation}
\label{lem1:1}
        (u-v)^\rT
        (x-y)
    \<
    -2
    r
    (|u|+|v|)
\end{equation}
holds for any $u \in N_S(x)$, $v\in N_S(y)$ in the corresponding inward normal cones (\ref{N}), where
\begin{equation}
\label{r}
    r
    :=
    \rho(c,\partial S)
\end{equation}
 is the distance (\ref{rho}) from the midpoint
\begin{equation}
\label{c}
     c:=\frac{1}{2}(x+y)
\end{equation}
(of the line segment with the endpoints $x$, $y$) to the boundary $\d S$.
Moreover, if $S$ is strongly convex with the constant $R$ (as in (\ref{R})),
then
\begin{equation}
\label{lem1:2}
        (u-v)^\rT
        (x-y)
    \<
    -
    \frac{1}{2R}
    (|u|+|v|)|x-y|^2.
\end{equation}
\hfill$\square$
\end{lem}
\begin{proof}
If the set $S$ is convex, then the closed ball $\overline{B}_r(c)$ of radius (\ref{r})  centred at the midpoint (\ref{c})
of the line segment with the endpoints $x,y\in S$, is contained
by $S$. In combination with the definition  (\ref{N}) of the inward normal cone, the
inclusion $\overline{B}_r(c)\subset S$ implies that for any $u\in N_S(x)$,
\begin{align*}
    0
    & \<
    \inf_{z\in S}
        u^\rT
        (z-x)
        \<
    \min_{z\in \overline{B}_r(c)}
        u^\rT
        (z-x)\\
    & =
    u^\rT(c-x)
    +
    \min_{w\in \overline{B}_r(0)}
        u^\rT w
        =
    \frac{1}{2}
        u^\rT(y-x)
    -r|u|,
\end{align*}
and hence,
\begin{equation}
\label{ineq1}
        u^\rT(x-y)
    \<
    -2r|u|.
\end{equation}
By a similar reasoning, for any $v \in N_S(y)$,
\begin{equation}
\label{ineq2}
        v^\rT(y-x)
    \<
    -2r|v|.
\end{equation}
By taking the sum of (\ref{ineq1}) and (\ref{ineq2}) and recalling
(\ref{r}), (\ref{c}), we arrive at
(\ref{lem1:1}).
Now, suppose the set $S$ is strongly convex. In this case,
(\ref{lem1:2}) cannot be obtained directly from (\ref{lem1:1}) by using the lower bound \cite{FO_1981}
\begin{equation}
\label{rR}
    r
    \>
    R
    -
    \sqrt{R^2-\frac{1}{4}|x-y|^2}
    \>
    \frac{1}{8R}
    |x-y|^2
\end{equation}
for the radius (\ref{r}) in terms of the strong convexity constant  $R$ of the set $S$,  because (\ref{rR}) leads to a more conservative factor $\frac{1}{4}$ than $\frac{1}{2}$ in (\ref{lem1:2}).
In order to prove (\ref{lem1:2}), suppose $x \in \partial S$ and $u \in
N_S(x)\setminus \{0\}$. Then, by using the unit inward
normal
\begin{equation}
\label{nu}
    \vartheta
    :=
    \frac{1}{|u|}u,
\end{equation}
it follows from (\ref{R}), applied to the set $S$, that
$
    S
    \subset
    \overline{B}_R(x+R\vartheta)
$. This inclusion is equivalent to
$
    |x+R\vartheta -y| \<  R
$
for all $y \in S$, and hence,
\begin{equation}
\label{ineq0}
    0\> |x+R\vartheta -y|^2 -  R^2 = |x-y|^2 +2R\vartheta^\rT (x-y),
\end{equation}
where use has been made of the property $|\vartheta| = 1$ which follows from (\ref{nu}). Therefore, multiplication of (\ref{ineq0}) by $|u|$ leads to
\begin{equation}
\label{2R1}
    2R
        u^\rT(x-y)
    \<
    -
    |u|
    |x-y|^2.
\end{equation}
Since the cone $N_S(x)$ reduces to the singleton $\{0\}$ for any interior
point $x\in \interior S$, the inequality (\ref{2R1}) remains valid for any
$x, y \in S$ and any $u\in N_S(x)$. Similarly, if $v \in N_S(y)$,
then
\begin{equation}
\label{2R2}
    2R
        v^\rT(y-x)
    \<
    -
    |v|
    |y-x|^2.
\end{equation}
The sum of (\ref{2R1}), (\ref{2R2}) yields (\ref{lem1:2}),
thus completing the proof.
\end{proof}

\section*{\bf B. A spectral bound for a class of symmetric matrices}
\label{sec:B}
\renewcommand{\theequation}{B\arabic{equation}}
\setcounter{equation}{0}

\begin{lem}
\label{lem:spec}
Let $a,b\in \mR$ and $g\in \mR^{n\x p}$. Then the largest eigenvalue of the real symmetric matrix
\begin{equation}
\label{M}
    M
    :=
    \begin{bmatrix}
      a I_n & g\\
      g^\rT & b I_p
    \end{bmatrix}
\end{equation}
is given by
\begin{equation}
\label{spec}
    \lambda_{\max}(M)
    =
    \frac{a+b}{2} + \sqrt{\|g\|^2 + \Big(\frac{a-b}{2}\Big)^2}.
\end{equation}
\hfill$\square$
\end{lem}
\begin{proof}
If $g=0$ in (\ref{M}), then $\lambda_{\max}(M) = \max(a,b) = \frac{a+b}{2} + \frac{|a-b|}{2}$, and (\ref{spec}) holds in this case. Now, let $g\ne 0$, in which case, $\lambda_{\max}(M)>\max(a,b)$. For any $\lambda >\max(a,b)$ (and hence, $\lambda \ne a$ and $\lambda \ne b$), the characteristic polynomial of the matrix (\ref{M}) admits the Schur complement decomposition  \cite{HJ_2007}
\begin{align}
\nonumber
    \det(\lambda I_{n+p} - M)
    & =
    \det
    \begin{bmatrix}
        (\lambda-a)I_n & - g\\
        - g^\rT & (\lambda - b)I_p
    \end{bmatrix}\\
\nonumber
    & =
    (\lambda-a)^n
    \det \Big((\lambda - b)I_p - \frac{1}{\lambda-a}g^\rT g\Big)\\
\nonumber
    & =
    (\lambda-a)^{n-p}
    \det (\sigma I_p - g^\rT g)    \\
\label{char}
    & =
    (\lambda-b)^{p-n}
    \det (\sigma I_n - g g^\rT),
\end{align}
where
\begin{equation}
\label{lamsig}
    \sigma:= (\lambda-a)(\lambda-b) =
    \lambda^2 - (a+b)\lambda + ab.
\end{equation}
In view of (\ref{char}), every such eigenvalue $\lambda$ of the matrix $M$ can be represented as a solution
\begin{equation}
\label{sigh}
    \lambda = \frac{a+b}{2} \pm \sqrt{\sigma + \Big(\frac{a-b}{2}\Big)^2}
\end{equation}
of the quadratic equation (\ref{lamsig}) for a common eigenvalue $\sigma$ of the matrices $g^\rT g$ and $gg^\rT$, whose spectra (consisting of the squared singular values of $g$) can differ from each other only by zeros \cite{HJ_2007}. The maximum in (\ref{sigh})  is achieved when the terms are summed and $\sigma=\lambda_{\max}(g^\rT g) = \lambda_{\max}(g g^\rT) = \|g\|^2$,
which establishes (\ref{spec}).
\end{proof}

Note that the right-hand side of (\ref{spec}) depends on $a$,  $b$ in a convex fashion, inheriting this property from the function $\lambda_{\max}$  for
real symmetric (or complex Hermitian) matrices.

\end{document}